\documentclass[12pt]{article}
\usepackage[margin=1in]{geometry}
\usepackage{amsmath, amsfonts, bm,amsthm}
\usepackage{algorithmic}
\usepackage[ruled,linesnumbered,lined]{algorithm2e}
\usepackage[latin1]{inputenc}
\usepackage{tikz}
\usepackage{hyperref}
\usetikzlibrary{shapes,arrows}
\usepackage{enumerate}
\usepackage{mdframed}
\usepackage{graphicx}
\usepackage{comment}
\usepackage{multirow}
\usepackage{verbatim}
\usepackage{authblk}

\newcommand{\xin}[1]{{\color{blue}{Xin: #1}}}

\makeatletter
\newtheorem*{rep@theorem}{\rep@title}
\newcommand{\newreptheorem}[2]{%
\newenvironment{rep#1}[1]{%
 \def\rep@title{#2 \ref{##1}}%
 \begin{rep@theorem}}
 {\end{rep@theorem}}}
\makeatother
\newreptheorem{theorem}{Theorem}
\newmdtheoremenv[innertopmargin=0pt]{conj}{Conjecture}

\newtheorem{definition}{Definition}

\newtheorem{lemma}{Lemma}
\newtheorem{theorem}{Theorem}
\newtheorem{claim}{Claim}

\makeatletter
\newtheoremstyle{indented}
  {3pt}
  {3pt}
  {\addtolength{\@totalleftmargin}{3.5em}
   \addtolength{\linewidth}{-3.5em}
   \parshape 1 3.5em \linewidth}
  {}
  {\bfseries}
  {.}
  {.5em}
  {}
\makeatother

\theoremstyle{indented}
\newtheorem{observation}{Observation}

\theoremstyle{definition}
\newmdtheoremenv[innertopmargin=0pt]{exmp}{Example}[section]

\newcommand{\mMS}{\mbox{MMS}} 
\newcommand{\mMSi}[1][i]{\mMS_{#1}} 
\newcommand{\amMS}[1]{$#1$-\mMS} 
\newcommand{\mmax}{maximin share} 

\newcommand{\inst}{\mathcal{I}}

\newcommand{\bundle}{B} 
\newcommand{\bundlei}[1][i]{\bundle_{#1}} 
\newcommand{\bundles}{\mathcal{\bundle}} 

\newcommand{\alloc}{A} 
\newcommand{\allocs}{\mathbf \alloc} 
\newcommand{\alloci}[1][i]{\alloc_{#1}} 
\newcommand{\alloclist}[1][n]{\left(\alloc_1,\dots,\alloc_{#1}\right)} 

\newcommand{\agents}{\mathcal{N}} 
\newcommand{\items}{\mathcal{M}} 

\newcommand{\va}{\mathcal{V}} 
\newcommand{\val}[1][i]{v_{#1}} 
\newcommand{\vai}[2][i]{v_{#1}\left(#2\right)} 
\newcommand{\valist}[1][n]{\left(\val[1],\dots,\val[#1]\right)}

\newcommand{\larg}[1]{\mathcal{L}\left(#1\right)}
\newcommand{\smal}[1]{\mathcal{S}\left(#1\right)}

\newcommand{\last}{\omega} 

\newcommand{\argmin}{\operatornamewithlimits{arg\ min}}

\newcommand{\nitest}[1]{\mathsf{NT}\left(#1\right)}

\newcommand{\cadin}[1]{\left\lvert#1\right\rvert}

\title{An Algorithmic Framework for Approximating Maximin Share Allocation of Chores}  

\author[1]{Xin Huang\thanks{xinhuang@campus.technion.ac.il }}
\affil{Technion - Israel Institute of Technology}
\author[2]{Pinyan Lu\thanks{lu.pinyan@mail.shufe.edu.cn}}
\affil{Shanghai University of Finance and Economics}

\date{} 

\begin{document}
\maketitle 
\begin{abstract}
We consider the problem of fairly dividing $m$ indivisible chores among $n$ agents.
The fairness measure we consider here  is the  maximin share.  The previous best known result is that there always exists a $\frac{4}{3}$-approximation maximin share allocation~\cite{barman2017approximation}.  With our algorithm, we can always find a  $\frac{11}{9}$-approximation maximin share allocation for any instance. We also discuss how to improve the efficiency of the algorithm and its connection to the job scheduling problem.
%
%
\end{abstract}

\section{Introduction}\label{sec-intro}

It is an important research theme in social science, economics and computer science to study how to allocate items among a number of different agents in a fair manner. The items could be something people like such as a house, cake or other resources which are called goods, or something people dislike such as tasks or duties which are called chores.  The rigorous study of the fair allocation problem dates back to the 1940's starting from the seminal work of~\cite{steinhaus1948problem}. At the beginning, researchers were interested in valuable divisible resources and nicknamed this problem ``Cake Cutting". Two well known fairness notions are defined and explored: 1) Envy freeness -- each agent prefers her own share of the cake over any other agent's~\cite{foley1967resource}; 2) Proportionality -- each agent gets a share  at least as valuable as the average of the whole cake~\cite{steinhaus1948problem}.


However, when researchers began to study the indivisible setting, the story changed dramatically. The reason is that, for indivisible items, envy-free and proportional allocations may not exist, and  even approximation is impossible.  A simple and somewhat awkward example is that of two agents with the same valuation for one single item. No matter how you allocate the item, there is one agent getting nothing. 

So how should we divide indivisible items fairly?  We need a more delicate fairness concept.  Budish~\cite{budish2011combinatorial} proposed a relaxation of proportionality, which is called maximin share. It is considered as  one of most successful fairness concepts for studying indivisible items. 
The idea of maximin share comes from the cut-and-choose protocol.  In the cut-and-choose protocol, there is a cutter dividing the whole set of resources  and the cutter must be the last one to choose her share. As the cutter may get the worst share, the maximin share for an agent  is the best value that the agent can guarantee if she is the cutter. An allocation is a \emph{maximin share allocation} if everyone gets a bundle at least as valuable as her maximin share.


Recently, the notion of maximin share attracted a lot of attention from the computer science community since the seminal work of~\cite{procaccia2014fair}. For goods, Kurokawa et al.~\cite{procaccia2014fair} show that a maximin share allocation may not exist. On the other hand, they demonstrate an exponential time algorithm for $\frac{2}{3}$ approximation of maximin share allocation. A line of works~\cite{DBLP:journals/talg/AmanatidisMNS17,barman2017approximation,ghodsi2018fair,DBLP:conf/soda/GargMT19,garg2019improved}  follow up to design efficient algorithms  and improve the approximation ratio. For chores, Aziz et al.~\cite{aziz2017algorithms}  show that maximin share allocation may not exist  and demonstrate a polynomial time algorithm for $2$ approximation. Later Barman and Khrishna Murthy
~\cite{barman2017approximation} improved this result  to a polynomial time algorithm for $\frac{4}{3}$ approximation.


So far, most efforts in this area are devoted to studying how to divide goods. The parallel problems for chores are less discussed in the community. We have two reasons to pay attention to indivisible chores setting.  1) Technically, it may betray some people's expectation that the problems for chores and goods are intrinsically different.  For example,  maximum Nash welfare allocation is  Pareto optimal with  certain fairness  for goods~\cite{caragiannis2016unreasonable,DBLP:journals/corr/BogomolnaiaMSY17}, however no single valued rule can be efficient with good fairness guarantee for chores~\cite{DBLP:journals/corr/BogomolnaiaMSY17}.
In this paper, we will show an algorithmic framework that is suitable for chores but has no direct implication for goods.
2) Practically, there are many applications in daily life in which allocating indivisible chores is involved. For example, household chores dividing, assignment of TA duties and the famous problem job scheduling etc. And this problem is also closely related to bin packing problem.

We focus on two questions: What is the best ratio for which an approximation allocation exist? Can we design an efficient algorithm with a better approximation ratio? We make significant contributions to these questions in this paper.





\subsection{Our results and techniques}
All of our results rely on a novel algorithmic framework which combines some existent ideas.
The first building block of our framework is a technique by Bouveret and Lema{\^\i}tre \cite{bouveret2016characterizing}, which allows us to focus on the class of  instances where agents share the same ordinal preference. This technique has been successfully applied to approximation of maximin share for goods~\cite{barman2017approximation,garg2019improved}. Interestingly, we identify a similarity between our algorithmic framework and well-known First Fit Decreasing (FFD) algorithm for the bin packing problem~\cite{johnson1973near}.
The core ideas of this part are simple: 1. As long as the bundle is within the bin size, add as many items as possible to the bundle;
2. Try to allocate large items first (as they're more problematic).

Under this algorithmic framework, we prove our main result.
\begin{reptheorem}{thm-exist}
For any chore division instance $\inst$, there always exists an $\frac{11}{9}$ approximation maximin share allocation.
\end{reptheorem}

To combine the above two ideas, the algorithm proceeds as follows: 
Order the chores in decreasing order (using the reduction from a general instance to an instance such that all agents share the same ordinary preference~\cite{bouveret2016characterizing}), and fill a bundle with as many chores as possible in a greedy fashion, as long as some agent thinks it is within
$\frac{11}{9}$ of her maximin share. Then one such agent takes the bundle and leaves the game. We repeat this $n$ times, where $n$ is the number of agents.
It is clear that each agent gets at most $\frac{11}{9}$ of her maximin share. To verify that it is indeed an $\frac{11}{9}$ maximin share allocation,
we only need to prove that all the chores are allocated during the process. This is however highly non-trivial.
We need to maintain some invariant property during the process and argue that it will allocate all the chores.

If one replaces  $\frac{11}{9}$ with a smaller ratio, can we still show that it will allocate all the chores? We do not know. We show by an example that one cannot make the ratio as small as $\frac{20}{17}$, but leave the tightness of the ratio as an interesting open question.

The above algorithm is quite simple, but needs to know the maximin share of each agent. Computing that value is precisely a job scheduling/load balancing problem, which is NP-hard. Since there is a PTAS for the job scheduling problem~\cite{jansen2016closing}, we have a PTAS for $\frac{11}{9}+\epsilon$ approximation of maximin share allocation.

%
%

However, the PTAS may not be considered as an efficient algorithm for some practical applications if $\epsilon$ is small.  To get a truly efficient algorithm, we notice that it may not be necessary to get an accurate value of maximin share. What we need is a reasonable lower bound of maximin share for each agent. 
This kind idea has already been applied to design efficient algorithms for fair allocation of indivisible goods~\cite{garg2019improved}. With this idea, we try to pre-allocate all chores in an appropriate and easy-to-implement way according to one particular agent's valuation, and then from this pre-allocation we can estimate a lower bound of maximin share for this particular agent. This complicates the argument, and it is not clear if we can get the same ratio of $\frac{11}{9}$. In this paper, we show that a slightly worse ratio of $\frac{5}{4}$ is achievable. We leave the problem of giving a polynomial time $\frac{11}{9}$ approximation algorithm as an open problem.

%
%
%
%

 One special case of our problem is that all agents have the same valuation for chores. We notice that the problem of job scheduling on identical machines is exactly this special case. Based on our algorithm,  we can design a very efficient $O(m\log m+n)$ time algorithm to get a  $\frac{11}{9}$ approximation of optimal scheduling. To the best of our knowledge, except for the PTAS which is not so efficient in practice,  there is no algorithm approximating optimal better than $\frac{4}{3}$.

%
%

\subsection{Related work}

The topic of fair division has a long history; it originates from the work of Steinhaus~\cite{steinhaus1948problem} in the 40s, which triggered vast literatures on this subject -- we refer the reader to the books~\cite{robertson1998cake,Moulin03} for an overview. Most of the literature in this area focuses on the divisible setting,
including very recent breakthroughs like the envy-free cake-cutting protocol of Aziz and McKenzie~\cite{AM16} and a similar result for chores~\cite{DBLP:conf/soda/DehghaniFHY18}. In contrast, fairly allocating indivisible items among agents has not been similarly popular, until very recently. The delayed interest is likely caused by the lack of suitable fairness notions.

The recent interest for the indivisible items setting was sparked by the definition of fairness notions that approximate envy-freeness and proportionality. In particular, the notions of  EF1 and EFX, defined by  Budish~\cite{budish2011combinatorial}  and Caragiannis et al.~\cite{caragiannis2016unreasonable} can be thought of as an approximate version of envy-freeness and have received much attention recently. Some works explore their existence~\cite{plaut2018almost,DBLP:conf/ec/CaragiannisGH19,DBLP:conf/soda/ChaudhuryKMS20,DBLP:conf/sigecom/ChaudhuryGM20}, others investigate the relationship with efficiency~\cite{caragiannis2016unreasonable,barman2018finding}.
Besides the concepts of EF1 and EFX mentioned above, approximate versions of envy-freeness include epistemic envy-freeness~\cite{ABCGL18} or notions that require the minimization of the envy-ratio~\cite{lipton2004approximately} and degree of envy~\cite{CEEM07,NR14} objectives.

The notion of maximin fair share (MMS) was first proposed by Budish~\cite{budish2011combinatorial} and inspired a line of works.  
In the seminal work of~\cite{procaccia2014fair}, the authors prove that MMS fairness may not exist, but a $\frac{2}{3}$ approximation of MMS  can be guaranteed. Due to their work, the best approximation ratio of MMS and a polynomial time algorithm for finding it becomes an intriguing problem.  A line of works~\cite{barman2017approximation,DBLP:conf/soda/GargMT19,DBLP:journals/talg/AmanatidisMNS17} tries to design an efficient algorithm for a $\frac{2}{3}$ approximation of MMS allocation.
Ghodsi et al. \cite{ghodsi2018fair} further improve the existence ratio of MMS to $\frac{3}{4}$. Shortly after, Grag and Taki~\cite{garg2019improved} show a polynomial time algorithm for $\frac{3}{4}$ approximation of MMS by combining all previous techniques for this problem.

For the chores setting, Aziz et al.~\cite{aziz2017algorithms} initiate the research on maximin share notion and provide a polynomial algorithm for 2 approximation. Utilizing the technique for goods, Barman et al.~\cite{barman2017approximation} also showed a polynomial time algorithm for $\frac{4}{3}$ approximation of MMS for chores. To the best of our knowledge, the ratio $\frac{4}{3}$ was the state of art before this paper.

Except maximin fairness, recently a line of study explores the problem fair division of indivisible chores from different perspectives. Aziz et al.~\cite{aziz2018fair} proposed a model for handling mixture of goods and chores. The paper~\cite{aziz2019strategyproof} showed
that strategyproofness would cost a lot on maximin fairness for chores. Aziz et al.~\cite{aziz2019weighted} considered the case that agents have different weights in the allocation process.

The problem of job scheduling/load balancing is a special case of allocating indivisible chores.  It is a fundamental discrete optimization problem, which has been intensely studied since the seminal work~\cite{graham1966bounds}. Graham~\cite{graham1969bounds} showed that the famous Longest Processing Time rule can give a $\frac{4}{3}$ approximation to the optimal. Later Hochbaum and Shmoys~\cite{hochbaum1987using} discovered a PTAS for this problem. A line of follow up works~\cite{alon1998approximation,jansen2010eptas,jansen2016closing} try to improve the running time by developing new PTAS algorithms.

Our algorithmic framework is similar to First Fit Decreasing (FFD) algorithm for the bin packing problem. Johnson in his doctoral thesis~\cite{johnson1973near} first showed the performance of FFD for the bin packing problem is tight to $\frac{11}{9}$ upon an additive error. To simplify the proof and tighten the additive error, subsequent works~\cite{baker1985new,yue1991simple,dosa2007tight} were devoted to this problem and finally got optimal parameters.   A modified and more refined version of FFD was proposed and proved to be approximately tight up to the ratio $\frac{71}{60}$~\cite{DBLP:journals/jc/JohnsonG85}.




\subsection{Organization}
In section \ref{sec-pre}, we introduce some basic notations and concepts for the paper. In section \ref{sec-frame}, we demonstrate our algorithmic framework which is the foundation of this work. We prove the existence of an $\frac{11}{9}$-approximation allocation by the algorithmic framework in section \ref{sec-exist}. Following the existence result, in section \ref{sec-poly} we push further to have an efficient polynomial time algorithm for \amMS{\frac{5}{4}} allocation. 
In section \ref{sec-special}, we connect our problem with the job scheduling problem and obtain an efficient algorithm.  Finally, in section \ref{sec-dis} we discuss some future directions and open problems for our algorithmic framework.

\section{Preliminary}\label{sec-pre}
We introduce some basic definitions and concepts for our model here.   An instance of the problem of dividing indivisible chores is denoted as $\inst=\langle\agents,\items,\va\rangle$, where $\agents$ is the set of agents, $\items$ is the set of chores, and $\va$ is the collection of all valuations.  For simplicity, we assume that the set of agents is $\agents=[n]$ and the set of chores is $\items=\{c_1,c_2,\dots,c_m\}$, where $n$ is the number of agents and $m$ is the number of chores. The  collection of all valuations $\va$ can be equivalently written  as $\valist$. For each agent $i$, the corresponding valuation function $\val:2^{\items}\rightarrow \mathbb{R}^+$ is additive, i.e., $\vai{S}=\sum_{c\in S}\vai{c}$ for any set  $S\subseteq\items$. 

Notice here we use non-negative valuation function which is the same as goods setting. However, the meaning is the opposite. Intuitively, the value is equivalent to the workload for $S$. So each agent wants to minimize her value. An allocation $\allocs=\alloclist$ is a $n$ partition of all chores $\items$ which  allocates all chores $\alloci$ to each agent $i$. We denote all possible allocations as the set $\Pi_n(\items)$.


The \mmax~
of an agent $i$ is defined as  $$\mMSi=\min_{\allocs\in\Pi_n(\items)}\max_{j\in[n]}v_i(\alloci[j]).$$

Here compare to maximin share, minimax share may be a more proper name. Since we minimum the value (duty/work load) that the agent can guarantee if she is the cutter. If one use negative value for chores, it can still be called maximin share, and we follow the literature to use the term ``maximin" and use positive value for notational simplicity. 


Following is a formal definition for \mmax~allocation.
\begin{definition} [Maximin share allocation]
 An allocation $\allocs\in\Pi_n(\items)$ is a \mmax~(\mMS) allocation  if $$\forall i\in\agents, \vai{\alloci}\le \mMSi.$$
\end{definition}
An allocation $\allocs$ is called \emph{\amMS{\alpha} allocation} if the inequality $\vai{\alloci}\le\alpha\cdot\mMSi$ holds for any agent $i$.

For the proof of our algorithm, the following is a useful definition.
\begin{definition}[Maximin share allocation for agent $i$]
 An allocation $\allocs\in\Pi_n(\items)$ is a \mmax~allocation for agent $i$,  if $$\forall j\in\agents, \vai{\alloci[j]}\le \mMSi.$$
\end{definition}

\begin{definition}[Identical ordinary preference]
An instance $\inst$ is called \emph{identical ordinary preference} (IDO) if there is a permutation $\sigma$ on $[m]$ such that, when $j\ge k$, we have the inequality $\vai{c_{\sigma(j)}}\ge\vai{c_{\sigma(k)}}$ holds for any agent $i$.
\end{definition}

As we will constantly use the notion of $j$-th largest chore of a bundle in the description of the algorithm and the proof, here we give a notation for it.

\begin{definition}\label{def-j-th}
Given an instance $\inst$ and a bundle (a set of chores) $\bundle\subseteq\items$, we denote by $\bundle[i,j]$   the $j$-th largest chore in bundle $\bundle$ from agent $i$'s perspective. For IDO instance, since every agent share the same ordinary preference, we will shortcut it as $\bundle[j]$.
\end{definition}

For example, for an IDO instance, the item $\items[1]$ and  the item $\items[2]$ would be the largest and the second largest chore of all.

\section{Algorithmic framework}\label{sec-frame}
In this section, we present a general algorithm framework for dividing chores. The algorithmic framework is building on a reduction and a heuristic. The reduction is from the work~\cite{bouveret2016characterizing}, which allows  us to focus on IDO instances. 
Then, we demonstrate a heuristic  for IDO instances, which is similar to  First Fit Decreasing algorithm for bin packing problem~\cite{johnson1973near}.


The reduction from general instances to IDO instances is captured by the following lemma. The original statement is for goods.  As the setting of chores  is slightly different from goods, we give a full detail in Appendix \ref{apd-sec-reduce} for the completeness.

\begin{lemma}\label{lem-general}
Suppose that there is an algorithm $G$  running in $T(n,m)$ time and returning an \amMS{\alpha} allocation for all identical ordinary preference instances. Then, we have an algorithm running in time $T(n,m)+O(nm\log m)$ outputting an \amMS{\alpha} allocation for all instances.
\end{lemma}



With the above reduction, we can focus on the identical ordinary preference. We introduce a heuristic which is a key part of our approximation algorithm.
The high level idea is that we setup a threshold for each agent, and then allocate large chores first and allocate them as much as possible with respect to the threshold.

\begin{algorithm}[H]
\KwIn{An IDO instance $\inst$, threshold values of agents $(\alpha_1,\dots,\alpha_n)$}
\KwOut{An allocation $\allocs$}
\BlankLine
	Let $R=\items$ and $T=\agents$ and $\allocs=(\emptyset,\dots,\emptyset)$\;
	\For(\tcp*[f]{Loop to generate bundles}\label{inalg-loop}){$k=1$ \KwTo $n$}
	{
		$\alloc=\emptyset$\;
		\For(\tcp*[f]{From the largest chore to the smallest}){ $j=1$ \KwTo  $|R|$}
		{
			\If{$\exists i\in T, \vai{\alloc\cup R[j]}\le \alpha_i$}
			{
				$\alloc=\alloc\cup R[j]$\;
			}	
		}
		$R=R\setminus \alloc$\;
		Let $i\in T$ be an agent  such that $\vai{\alloc}\le\alpha_i$\;\label{inalg-select}
		$\alloci=\alloc$ and $T=T\setminus i$\;
	}
	\Return {$\allocs$}
\caption{Algorithm for \amMS{\alpha}}\label{alg-approx}	
\end{algorithm}

For this algorithm we have the following observation.

\begin{lemma}\label{lem-not-left}
Suppose that the inequality $\alpha_i\le\alpha\cdot\mMSi$ holds for each $i$. If all chores are allocated by the algorithm, then the allocation returned by Algorithm \ref{alg-approx} is an \amMS{\alpha} allocation.
\end{lemma}

To analyze the algorithm, we only need to focus on what threshold values $(\alpha_1,\dots,\alpha_n)$ will make the algorithm allocate all chores.

\section{Main result}\label{sec-exist}
In this section, we show that an $\frac{11}{9}$-approximation \mmax~allocation always exists by our algorithmic framework. 
To achieve this goal, it is sufficient to consider IDO instances solely (see Lemma \ref{lem-general}). For simplicity, we only deal with IDO instances for all proofs in this section.


\begin{theorem}\label{thm-exist}
For any chore division instance $\inst$, there always exists an \amMS{\frac{11}{9}} allocation.
\end{theorem}

By Lemma \ref{lem-not-left}, it is sufficient to prove that on inputting threshold value $\alpha_i=\frac{11}{9}\cdot \mMSi$ to Algorithm \ref{alg-approx}, all chores will be allocated.
To prove that, we focus on the agent who gets the last bundle in the execution of Algorithm \ref{alg-approx} (the agent  was chosen in line \ref{inalg-select} of Algorithm \ref{alg-approx},  when $k=n$ in the loop of line \ref{inalg-loop}). Let $\last$ be a special index for this agent. The whole proof will take the view from agent $\last$'s perspective. If we can prove that all remaining chores will be allocated to agent $\last$, then the correctness of theorem is implied.  For the simplicity of the presentation, we assume  that $\mMSi[\last]=1$.

Our proof has two parts: First we prove that all ``small" chores will be allocated, and then in the second part we analyze what happened to ``large" chores. 

Here we give a formal definition of ``small" and ``large". Notice that these notions will also be used in other sections.
\begin{definition}\label{def-big-little}
With a parameter $\alpha$ and an index $i\in\agents$, we define the set of ``large" chores as  $$\larg{i,\alpha}=\{c\in\items\mid\vai{c}>\alpha\}$$ and define the set of ``small" chores as  $$\smal{i,\alpha}=\{c\in\items\mid\vai{c}\le\alpha\}.$$
\end{definition}

First we prove that, from the last agent $\last$'s view, there is no small chore remained after the algorithm terminates.
\begin{lemma}\label{lem-no-less-2/9}
If we input threshold value $\alpha_i=\frac{11}{9}\cdot \mMSi$ to Algorithm \ref{alg-approx}, then all chores in the set $ \smal{\last, 2/9}$ will be allocated after the algorithm terminates.
\end{lemma}
\begin{proof}
We will prove the statement by contradiction. 
 Suppose that there was a chore $c\in \smal{\last, 2/9}$ remained. For any bundle $\alloci$, we have the inequality  $\vai[\last]{\alloci\cup c}>11/9$, otherwise Algorithm \ref{alg-approx} would allocate the chore $c$ to the bundle $\alloci$. As the valuation $\vai[\last]{c}\le 2/9$, we get the inequality $\vai[\last]{\alloci}>1$ for all $i\in\agents$. Recall that the value  $\mMSi[\last]$ is 1 by our assumption. So the total valuation should not exceed $n$. When we add up the valuations of all bundles, we have $\sum_{i\in\agents}\vai[\last]{\alloci}>n$, which is a contradiction. Therefore, there is no chore $c\in \smal{\last, 2/9}$ remained.
\end{proof}


Next we will analyze the allocation of large chores. Notice that the algorithm always gives the priority of large chores over small chores. In other words, no matter what kind of small chores we have and how they are allocated, they do not influence the allocation of large chores. Thus, we can focus on large chores without considering any small chores. 

Before presenting the details of the proof, here we give a high level idea of the proof. Suppose that $\allocs$ is the allocation returned by Algorithm \ref{alg-approx} and $\bundles$ is a \mmax~allocation for agent $\last$. We try to imagine that the allocation $\allocs$  is generated round by round from the allocation $\bundles$. In each round $k$, we swap some chores so that bundle $\bundlei[k]$ equals  bundle $\alloci[k]$, and then we will not touch this bundle anymore.  
Our goal is to prove that during the swapping process no bundle becomes too large. Then, when we come to the last bundle, agent $\last$ could take all remaining chores. To achieve this goal, we choose proper parameters and carefully specify the swap operation so that only two types of swaps are possible. 

Suppose that in round $k$, we swap chores between bundle $\bundlei[k]$ and bundle $\bundlei[j]$ where $j>k$. The target of the swapping is to make bundle $\bundlei[k]$ close to bundle $\alloci[k]$, and meanwhile bundle $\bundlei[j]$ will not increase too much. See Figure \ref{fig-type}.
\begin{itemize}
\item  
Type 1:  We will swap  chore $c_j\in\bundlei[j]$ with chore $c_k\in\bundlei[k]$ such that $\vai[\last]{c_j}\ge\vai[\last]{c_k}$ (possibly $c_k=\emptyset$). 
After this swapping, the valuation of the bundle $\bundlei[j]$ would not increase.  
\begin{figure}
\centering
\includegraphics[scale=0.4]{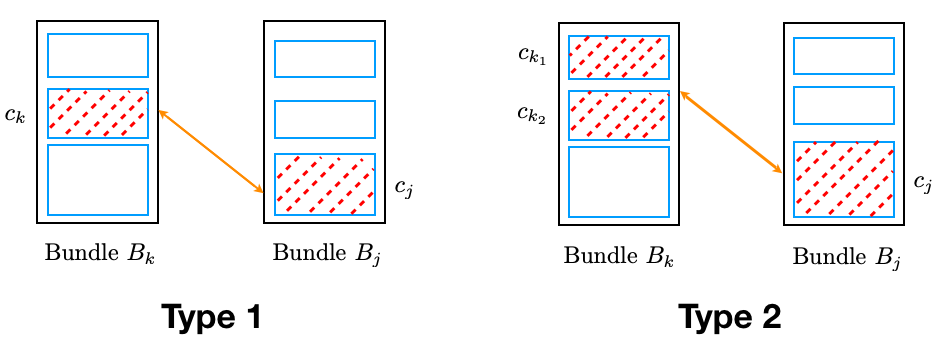}
\caption{Two types of swaps. Blue squares are chores. Red doted squares are chores for swapping.}
\label{fig-type}
\end{figure}
\item  Type 2: We will swap chore $c_j\in\bundlei[j]$ with two  chores $c_{k_1},c_{k_2}\in\bundlei[k]$. We have that $\vai[\last]{c_{k_1}}\le\vai[\last]{c_j}$ and $\vai[\last]{c_{k_2}}\le\vai[\last]{c_j}$. However, it is possible that  $\vai[\last]{\{c_{k_1},c_{k_2}\}}>\vai[\last]{c_j}$.  This swap could make the bundle $\bundlei[j]$ larger. By a delicate choice of the parameter, this kind of increment can be well upper bounded and would not accumulate, i.e., there is at most one Type 2 swap involving  bundle $\bundlei[j]$ for any $j$. 
\end{itemize}

\begin{lemma}\label{lem-no-larger-2/9}
If we input threshold value $\alpha_i=\frac{11}{9}\cdot \mMSi$ to Algorithm \ref{alg-approx}, then after the algorithm terminates, all chores of the set $ \larg{\last,2/9}$ will be allocated.
\end{lemma}
\begin{proof}
Suppose that $\allocs$ is the allocation returned by the algorithm and $\bundles$ is a \mmax~allocation for the last agent $\last$. We try to swap chores in \mmax~allocation $\bundles$ so that it will become the allocation $\allocs$. The proof is to analyze what happened in this swapping process. 

Recall that the value  $\mMSi[\last]$ is 1 by our assumption. 
We only care about large chores here.  
 So let us define the bundle $\alloci^*$ as $\alloci\cap \larg{\last,2/9}$ and the bundle  $\bundlei^*$ as $\bundlei\cap \larg{\last,2/9}$.
For the simplicity of the proof, we assume that all bundles are ordered by the largest chore in each bundle, i.e., $\alloci[i]^*[1]\ge \alloci[j]^*[1]$  for all $i<j$. Notice that for the allocation $\allocs$, this order is exactly the order from the  for-loop in line \ref{inalg-loop} of Algorithm \ref{alg-approx}.


We start with bundle collection $\bundles^{(0)}=\left(\bundlei[1]^*,\dots,\bundlei[n]^*\right)$  for our swapping process. Bundle collection $\bundles^{(0)}$ is an allocation of large chores such that $\vai[\last]{\bundlei[i]^*}\le 1$ for all $i\in\agents$. There are $n$ rounds  of the swapping process in total. In each round $t$, we generate a bundle collection $\bundles^{(t)}$ by swapping some chores in bundle collection $\bundles^{(t-1)}$, so that  bundle $\bundlei[t]^{(t)}=\alloci[t]^*$. Notice that we do not need to touch bundles with index strictly less than $t$ to fix bundle $\bundlei[t]^{(t)}$.  So we have $\bundlei^{(t)}=\alloci^*$ for all $i\le t$.

If this swapping process could be successful to the round $n$, then all large chores are allocated in allocation $\allocs$, which implies the lemma.   To show this, we carefully specifiy the swapping process so that bundle collection $\bundles^{(t)}$ satisfies a good property. 


To introduce the good property, we need some definitions. Let $$U(\bundle)=\min_{c\in \bundle} \vai[\last]{c}+\max_{c\in\bundle}\vai[\last]{c}$$ denote the sum of largest and smallest chore in the bundle.

\bigskip
{\bf Good bundle:} A bundle $\bundle$ is \emph{good} if 1) $\vai[\last]{\bundle}\le1$ or, 2) $U(\bundle)<5/9\text{ and }|\bundle|=4.$ 
\bigskip

The second condition of the good bundle definition is related to the Type 2 swap which is mentioned in the high level idea. This condition can help us bound and control the damage which is caused by the second type swap. 

Here we slightly abuse the concept ``good" for bundle collection.

\bigskip
{\bf Good bundle collection:}  A bundle collection $\bundles^{(t)}$ is \emph{good} if $\bundlei^{(t)}=\alloci^*$ for all $i\le t$, and  each bundle $\bundlei^{(t)}$ is good for $i>t$.
\bigskip

Clearly, the bundle collection $\bundles^{(0)}$ is good. 
This would be the basis for our induction. Now we need to describe the swapping process.
To  do it precisely, we introduce two operations $SWAP$ and $MOVE$, where the operation $SWAP$ would swap two chores and the operation $MOVE$ would move a set of chores  from their original bundles to a target  bundle.

Given a bundle collection $\bundles$ and two chores $c_1$ and $c_2$, the $SWAP$ operation will generate a new bundle collection $\bundles^{\#}=SWAP(\bundles,c_1,c_2)$ such that: If $c_1=c_2$ or, $c_1$ and $c_2$ are in the same bundle, then $\bundles^{\#}=\bundles$; otherwise, we swap these two chores.
\begin{equation*}
 \bundlei^{\#}=
  \begin{cases}
   \bundlei & \text{if } c_1,c_2\notin \bundlei  \\
   (\bundlei\setminus c_1)\cup c_2       & \text{if } c_1\in\bundlei \\
   (\bundlei\setminus c_2)\cup c_1  & \text{if } c_2\in\bundlei
  \end{cases}
\end{equation*}

Given a bundle collection $\bundles$, an  index $j\in\agents$ and a subset of chores $T\subseteq\items$, the $MOVE$ operation will generate a new bundle collection $\bundles^{\#}=MOVE(\bundles,j,T)$ such that
 \begin{equation*}
  \bundlei^{\#}=
   \begin{cases}
    \bundlei\setminus T & \text{if } i\neq j \\
    \bundlei\cup T      & \text{if } i=j \\
   \end{cases}
 \end{equation*}
 Notice that chores set $T$ could contain chores from different bundles. 


Now we describe and prove the correctness of  the swapping process by induction. Suppose that bundle collection $\bundles^{(k)}$ is good. We will swap and move some chores between bundle $\bundlei[k+1]^{(k)}$ and bundles with index greater than $k+1$, so that the bundle $\bundlei[k+1]^{(k)}$ becomes bundle $\alloci[k+1]^*$.  Without loss of generality, we can assume that the largest chore in  bundle $\bundlei[k+1]^{(k)}$ and  bundle $\alloci[k+1]^*$ are the same, i.e., chore $\bundlei[k+1]^{(k)}[1]=\alloci[k+1]^*[1]$. This is the largest chore among  the set $\bigcup_{i\ge k+1} \bundlei^{(k)}$ by our assumption of the order of allocation $\allocs$ (The assumption made at the beginning of the whole proof).   
 By the condition of good bundle collection and the value of large chores, we know that  bundle $\bundlei[k+1]^{(k)}$ contains at most $4$ chores.
 
 Before the case analysis, here we give a  claim which can help us combine several cases into one. 
\begin{claim}\label{claim-4chores}
For the cases $\cadin{ \bundlei[k+1]^{(k)}}=\mbox{1, 2, and }4$, we have the inequalities $\cadin{\alloci[k+1]^*}\ge \cadin{\bundlei[k+1]^{(k)}}$ and $\vai[\last]{\bundlei[k+1]^{(k)}[j]}\le\vai[\last]{\alloci[k+1]^*[j]}$ for $j\le\cadin{\bundlei[k+1]^{(k)}}$.
\end{claim}
 Please read Appendix \ref{apd-no-2/9} for the proof of Claim \ref{claim-4chores}.
 
  Based on the size of bundle $\bundlei[k+1]^{(k)}$, we have the following case analysis.

\begin{itemize}
\item $\cadin{\bundlei[k+1]^{(k)}}=1,2,\mbox{ and } 4$:  Let $d=\cadin{\bundlei[k+1]^{(k)}}$ be the cardinality of the bundle. 
By Claim \ref{claim-4chores}, we have  $\cadin{\alloci[k+1]^*}\ge \cadin{\bundlei[k+1]^{(k)}}$ and $\vai[\last]{\bundlei[k+1]^{(k)}[j]}\le\vai[\last]{\alloci[k+1]^*[j]}$ for all $j\le d$.   We swap each pair of chores $\bundlei[k+1]^{(k)}[j]$ and $\alloci[k+1]^*[j]$ for $j\le d$. Here all swaps are Type 1 (described in high level idea) swap. 
 If the cardinality $\cadin{\alloci[k+1]^*}>d$, then we just move the remaining chores  to bundle $\bundlei[k+1]^{(k)}$.  Please see Figure \ref{fig-case124} for this process.

Mathematically, let bundle collection $G^{(1)}=\bundles^{(k)}$. The swapping of the first $d$ chores  could be represented as the following form  $G^{(j)}=SWAP\left(G^{(j-1)},\bundlei[k+1]^{(k)}[j],\alloci[k+1]^*[j]\right)$ for $2\le  j\le d$. The bundle collection $G^{(d)}$ is the resulting bundle after these swaps.  And the moving operation to generate bundle collection $\bundles^{(k+1)}$ is equivalent to 
$$\bundles^{(k+1)}=MOVE\left(G^{(d)},k+1,\alloci[k+1]^*\right).$$ 

It is not hard to see after all these operations, no bundle with index larger than $k+1$  becomes larger nor gets more chores. 
 A good bundle will still be a good bundle. Therefore, when we get the  bundle collection $\bundles^{(k+1)}$, it is a good bundle collection. 
 
 \begin{figure}[h]
\centering
\includegraphics[scale=0.4]{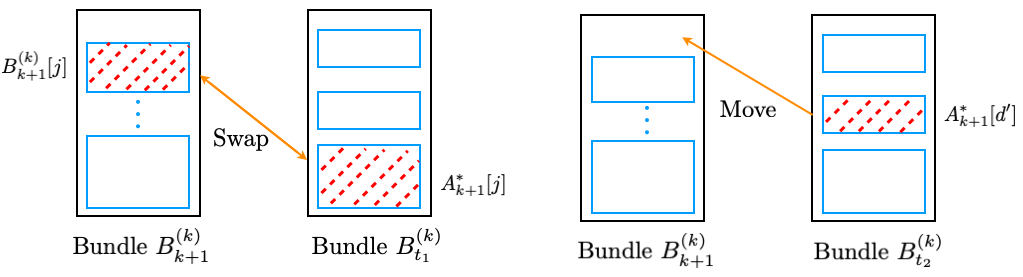}
\caption{When $j\le d$, the operation would be a Type 1 swapping (type is described in high level idea). For the case $d'>d$, the operation would be  the move operation.}
\label{fig-case124}
\end{figure}

\item $\cadin{\bundlei[k+1]^{(k)}}=3$: As chore $\bundlei[k+1]^{(k)}[1]$ equals chore $\alloci[k+1]^*[1]$ and $\vai[\last]{\left\{\bundlei[k+1]^{(k)}[1],\bundlei[k+1]^{(k)}[2]\right\}}<11/9,$ we can at least allocate chore $\bundlei[k+1]^{(k)}[2]$ with chore  $\alloci[k+1]^*[1]$ in the algorithm.
Thus, we have $\cadin{\alloci[k+1]^*}\ge 2$. Let $$G=SWAP\left(\bundles^{(k)},\bundlei[k+1]^{(k)}[2],\alloci[k+1]^*[2]\right)$$ be the bundle collection that  swap chore $\bundlei[k+1]^{(k)}[2]$ and chore $\alloci[k+1]^*[2]$.  Now we consider the following inequality:
\begin{equation}\label{ineq-2-3}
\vai[\last]{\left\{\alloci[k+1]^*[1],\alloci[k+1]^*[2], \bundlei[k+1]^{(k)}[3]\right\}}\le 11/9
\end{equation}

If the Inequality \ref{ineq-2-3} holds,  then we have $\cadin{\alloci[k+1]^*}\ge 3$ and $\vai[\last]{\bundlei[k+1]^{(k)}[3]}\le\vai[\last]{\alloci[k+1]^*[3]}$. For this case, we will swap chore $\bundlei[k+1]^{(k)}[3]$ and chore $\alloci[k+1]^*[3]$ and move remaining chores in bundle $\alloci[k+1]^*$ to bundle $\bundlei[k+1]^{(k)}$, i.e.,
$$\bundles^{(k+1)}=MOVE\left(SWAP\left(G,\bundlei[k+1]^{(k)}[3],\alloci[k+1]^*[3]\right),k+1,\alloci[k+1]^*\right).$$
The operations for this subcase are exactly the same situation as above (see Figure \ref{fig-case124}). 
By a similar argument,  bundle collection $\bundles^{(k+1)}$ is good.

If the Inequality \ref{ineq-2-3} does not hold, then we move chore $\bundlei[k+1]^{(k)}[3]$ to the bundle where chore $\alloci[k+1]^*[2]$ comes from. And then move  other chores in bundle $\alloci[k+1]^*$  to bundle $\bundlei[k+1]^{(k)}$. Please see Figure \ref{fig-case3} for this process.

 \begin{figure}[h]
\centering
\includegraphics[scale=0.4]{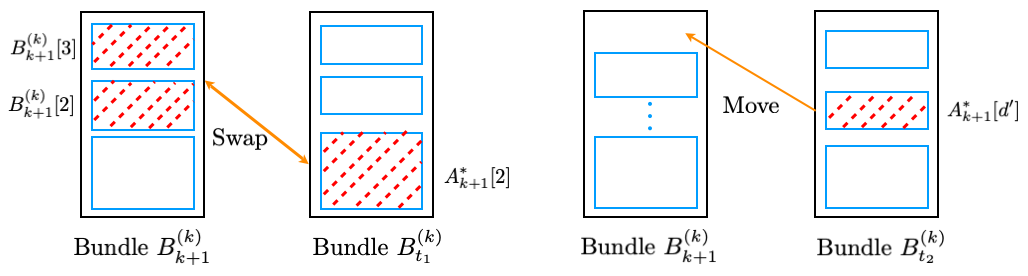}
\caption{There would be one operation of the Type 2 swap for chores $\alloci[k+1]^*[2],\bundlei[k+1]^{(k)}[2]$ and $\bundlei[k+1]^{(k)}[3]$. For the case $d'>2$, the operation would be the move operation.}
\label{fig-case3}
\end{figure}
Let $j$ be an index of a bundle such that $\alloci[k+1]^*[2]\in\bundlei[j]^{(k)}$. Then, formally this process is
$$\bundles^{(k+1)}=MOVE\left(MOVE\left(G,j,\bundlei[k+1]^{(k)}[3]\right),k+1,\alloci[k+1]^*\right).$$

This is the most complicated case since bundle $\bundlei[j]^{(k+1)}$ contains more items than bundle $\bundlei[j]^{(k)}$. 
We prove that  bundle $\bundlei[j]^{(k+1)}$ is good in Claim \ref{claim-two-one}. And for other bundles with index larger than $k+1$, it either does not change or removes some chores.  Therefore, bundle collection $\bundles^{(k+1)}$ is good.



\end{itemize}

\begin{claim}\label{claim-two-one}
 Let $j$ be the index such that the chore $\alloci[k+1]^*[2]\in\bundlei[j]^{(k)}$. For the case $\cadin{\bundlei[k+1]^{(k)}}=3$ and Inequality \ref{ineq-2-3} does not hold, 
after all $SWAP$ and $MOVE$ operations, bundle $\bundlei[j]^{(k+1)}$ is good.
\end{claim}
Please read Appendix \ref{apd-no-2/9} for the proof of Claim \ref{claim-two-one}.

By this case analysis, we show that the bundle collection will keep the good property until the last one $\bundles^{(n)}$. So we prove the lemma.

\end{proof}

Using the  above two lemmas, now we can show the  correctness of theorem \ref{thm-exist}.

\begin{proof}[Proof of Theorem \ref{thm-exist}]

Combing Lemma \ref{lem-no-less-2/9} and Lemma \ref{lem-no-larger-2/9}, we have that Algorithm \ref{alg-approx} will output a \amMS{11/9} allocation with input such that $\alpha_i=\frac{11}{9}\cdot \mMSi$.
\end{proof}

We can easily transfer the existence result of Theorem \ref{thm-exist} to a polynomial  approximation scheme. The only computational hardness part of our algorithmic framework is the valuation of \mmax~for each agent. 
Notice that computing \mmax~of an agent is exactly the makespan of a job scheduling problem, which is a famous NP-hard problem. 
This can be solved by a PTAS from job scheduling literature~\cite{jansen2016closing}. Therefore, from the result of existing 11/9 approximation \mmax~allocation, we have a PTAS for $11/9+\epsilon$ approximation \mmax~allocation. The constant 11/9 could be improved if anyone can prove a better existence ratio of our algorithmic framework.

Given the above result, it is natural to ask what is the best approximation ratio of our algorithmic framework. Though we cannot prove the best ratio now, we present the following example to show a lower bound of our technique.
\begin{exmp}\label{exm-lower-bound}
In this example, we consider an instance of 14 chores and 4 agents with an identical valuation (not just ordinal, but cardinal is also same). The valuation of each chore is demonstrated by the \mmax~allocation.

A \mmax~allocation of this instance is that $$\alloci[1]=\left\{\frac{9}{17},\frac{4}{17},\frac{4}{17}\right\}, \alloci[2]=\left\{\frac{7}{17},\frac{6}{17},\frac{4}{17}\right\}\text{ and }\alloci[3]=\alloci[4]= \left\{\frac{5}{17},\frac{4}{17},\frac{4}{17},\frac{4}{17}\right\},$$ where the numbers are valuations of each chore. Obviously, the \mmax~of each agent is 1. 

For any $\alpha$ such that $1\le \alpha<20/17$, if we input  threshold values $(\alpha,\alpha,\alpha,\alpha)$  to Algorithm \ref{alg-approx}, then we will get the first bundle $\alloci[1]^*=\{\frac{9}{17},\frac{7}{17}\}$, the second bundle $\alloci[2]^*=\{\frac{6}{17},\frac{5}{17},\frac{5}{17}\}$ and the third bundle $\alloci[3]^*=\{\frac{4}{17},\frac{4}{17},\frac{4}{17},\frac{4}{17}\}$. When we come to the last bundle, the total valuation of remaining chores is $\frac{20}{17}$. Since the threshold $\alpha<20/17$, we cannot allocate all chores.
\end{exmp}

Example \ref{exm-lower-bound} implies the following result.
\begin{theorem}\label{thm-lower-bound}
Algorithm \ref{alg-approx} cannot guarantee to find an \amMS{\alpha} allocation when $\alpha<20/17$.
\end{theorem}

\section{An efficient algorithm for practice}\label{sec-poly}
To run our algorithm for $\frac{11}{9}$-approximation maximin share allocation, we need to know the \mmax~of each agent. The computation of the maximin share is NP-hard. Though we have a PTAS for $\left(\frac{11}{9}+\epsilon\right)$-approximation, even if we want to get a $\frac54$-approximation with currently best PTAS for job scheduling~\cite{jansen2016closing}, the running time could be more than $2^{39000}+poly(n,m)$. This is not acceptable for a real computation task. 
To attack this computational embarrassment, we give a trial on designing an efficient polynomial time algorithm for $\frac54$-approximation \mmax~allocation when valuations are all integers.

The design of the efficient algorithm relies on an observation: It is not necessary to know the exact value of \mmax. A reasonable lower bound of \mmax~could serve the same end. This kind of observation was successfully applied to the goods setting~\cite{garg2019improved}.

From this observation, the basic idea of the design is to find a good lower bound of \mmax. And from the lower bound, we compute a  proper threshold value for executing  Algorithm \ref{alg-approx}. To make this idea work, we have two problems to solve: 1) How to find a proper lower bound of \mmax? 2) How to use this lower bound to get a good threshold value for the algorithm? 

We find that a threshold testing algorithm could be an answer to these problems. The threshold testing algorithm tries to allocate relative large chores in a certain way to see whether the threshold is large enough for an agent or not. 
For the first problem, we can use the threshold testing algorithm as an oracle to do binary search for finding a reasonable lower bound of \mmax~of each agent. For the second problem, 
the allocation generated by the threshold testing algorithm could serve as a benchmark to help us find a good threshold value.

\subsection{Difficulties on threshold testing}\label{subsec-failure}
To explain some key points of the threshold testing, we introduce a simple trial -- \emph{naive test}  which is  building on Algorithm \ref{alg-approx} straightforwardly.  The naive test gives a  polynomial-time $\frac{11}{9}$-approximation algorithm for the job scheduling problem (see Section \ref{sec-special}).  However, we construct a counter example which is an IDO instance for the naive test in this section. 


Basically, the naive test tries to run Algorithm \ref{alg-approx} directly on a single agent. 
Suppose that we have an instance $\inst=\langle\agents,\items,\va\rangle$.  Let $\nitest{i,s_i}$  denote a naive test, where $i$ is an agent  and $s_i$ is a threshold. The naive test $\nitest{i,s_i}$ will first construct an instance $ID=\langle\agents,\items,\va'\rangle$ such that $\va'=(\val,\dots,\val)$, i.e., valuations are all same as $\val$. Then it will run Algorithm \ref{alg-approx} with the $ID$ instance and threshold values $\left( s_i,\dots, s_i\right)$. If all chores are allocated in Algorithm \ref{alg-approx}, then the naive test $\nitest{i,s_i}$ returns ``Yes", otherwise returns ``No".

 Given an agent $i$, let $s^*_i=\min_{s_i}\left\{s_i\mid\nitest{i,s_i} \text{returns ``Yes"}\right\}$ be the minimal value passing the naive test. 
 By the proof of Theorem \ref{thm-exist}, it is not hard to see that the value $s^*_i$ satisfies $s^*_i\le\frac{11}{9}\cdot\mMSi$.  
Based on this observation, it is natural to try the following.

\bigskip
{\bf Trial approach}: First, compute the minimal threshold $s^*_i$ for each agent.\footnote{How to compute the value of $s^*_i$ is not important for our discussion on the properties of threshold testing.}  Then run Algorithm \ref{alg-approx} with thresholds $\{ s^*_i\}_{i\in\agents}$ to get an allocation.
\bigskip

If the following conjecture is true, then this approach would work. 

\begin{conj}\label{conj-mon}
Montonicity: If Algorithm \ref{alg-approx} can allocate all chores with threshold values $\{\alpha_i\}_{i\in\agents}$, then  Algorithm \ref{alg-approx} should allocate all chores with any threshold values $\{\beta_i\}_{i\in\agents}$  such that $\beta_i\ge\alpha_i$ for all $i$. 
 \end{conj}
For the naive test, this conjecture implies that if the test $\nitest{i,s_i}$ returns ``Yes", then the test $\nitest{i,s'_i}$ returns ``Yes" for any $s'_i\ge s_i$.
Unfortunately, Conjecture \ref{conj-mon} is not true even for the naive test (all agents have the same valuation). 
Here we give a counter example.
\begin{exmp}\label{exm-no-monoton}

In this example, we consider an instance of 17 chores and 4 agents. All agents have the same valuation. The valuation of each chore is demonstrated by the \mmax~allocation. 

The \mmax~allocation of this instance is that $$\alloci[1]=\{5.1,1.2,1.2\}, \alloci[2]=\alloci[3]=\{2.75,2.75,1,1\}\text{ and }\alloci[4]=\{2.5,1,1,1,1,1\}.$$ The numbers are valuations of each chore.  The \mmax~and value $s^*_i$ of each agent is 7.5. It is easy to verity that all chores will be allocated if we input the threshold values $(7.5, 7.5, 7.5, 7.5)$ to Algorithm \ref{alg-approx}. And the allocation is exactly the  allocation we give here. 

However, if we input threshold values $(7.6,7.6,7.6,7.6)$ to Algorithm \ref{alg-approx}, we will have $$\alloci[1]=\{5.1,2.5\}, \alloci[2]=\alloci[3]=\{2.75,2.75,1.2\}\text{ and }\alloci[4]=\{1\}\times 7.$$ There are 2 chores remained. 

So for this example, naive test $\nitest{i,7.5}$ will return ``Yes", but $\nitest{i,7.6}$ will return ``No".
\end{exmp}
This example reveals a surprising fact about Algorithm \ref{alg-approx}. When we have more spaces to allocate chores, it could be harmful. The connection between Conjecture \ref{conj-mon} and the trial approach is that: In terms of  the last agent's perspective, when we do the trial approach,  the effect of heterogeneous valuations would act as enlarging threshold. Next we construct an example based on Example \ref{exm-no-monoton} to show how the trial approach failed. 

\begin{exmp}\label{exm-fail}
We construct an IDO instance consists of 4 agents and 17 chores. For those 4 agents, three of them (agents $t_1$, $t_2$ and $t_3$) have exact the same valuation as Example \ref{exm-no-monoton}. Agent $t_4$ is a special agent. The valuation of agent $t_4$ is demonstrated by the \mmax~allocation: 
$$\alloci[1]=\{5.1,2.4\}, \alloci[2]=\alloci[3]=\{2.75,2.75, 2\}\text{ and }\alloci[4]=\{5/6\}\times 9.$$ 
By our construction, it is easy to verify that the value $s^*_i$ equals  7.5 for any agent $i$.  The trial approach will  input the threshold (7.5,7.5,7.5,7.5) to Algorithm \ref{alg-approx}. 

Notice, this is an IDO instance. For agents $t_1$, $t_2$ and $t_3$, the largest 4 chores are valued at $\{5.1,2.75,2.75,2.5\}$. For agent $t_4$, the largest 4 chores are valued at $\{5.1,2.75,2.75,2.4\}$.  When we do the trial approach,  agent $t_4$ will get the first bundle, which is $\alloc=\{5.1,2.4\}$ for agent $t_4$.  And this bundle contains chores $\{5.1,2.5\}$ in terms of the valuation of other agents. Just like Example \ref{exm-no-monoton}, there will be 2 chores unallocated at the end. So the trial approach fails. 
\end{exmp}





\subsection{Threshold testing}
Before we present the threshold testing algorithm, let us summarize the issue of the naive test by Figure \ref{fig-axis}. The naive test will return a threshold value $\alpha$ (a red value in Fig \ref{fig-axis} between $\mMSi$ and $\frac{11}{9}\mMSi$). When we run Algorithm \ref{alg-approx}, the effect of other agents can be viewed as increasing the threshold $\alpha$. Let us call the threshold after increasing \emph{effective threshold}. If the effective threshold is a blue value, then Algorithm \ref{alg-approx} could be failed. To resolve this situation, we need to find a threshold $\alpha'$ such that every value  larger than $\alpha'$ will pass the test.

  \begin{figure}[h]
\centering
\includegraphics[scale=0.4]{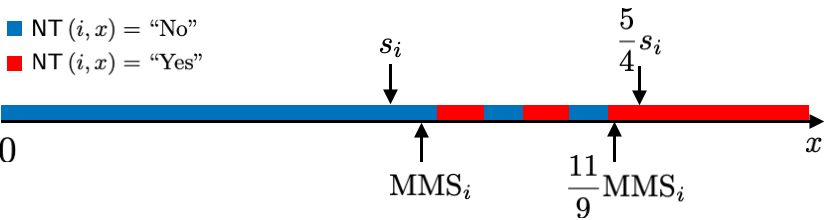}
\caption{The $x$-axis is the threshold value for the naive test. Blue means the value does not pass the test. Red means the value pass the test.}
\label{fig-axis}
\end{figure}

From this idea, we try to find a  value $s_i$, which is a  lower bound on $\mMSi$. Then scale it large enough ($\frac{5}{4}\cdot s_i$) so that no larger value will fail the test. This can be done by allocating chores with valuation larger than $s_i/4$ in a special way.

Here is a detailed description of our threshold testing algorithm.

\begin{algorithm}[H]
\KwIn{Index $i$ of an agent, a threshold value $s_i$}
\KwOut{Yes/No}
\BlankLine
	 Let $k=|\larg{i,s_i/2}|$ \tcp*[l]{Please refer to  Definition \ref{def-big-little}}
	 Initialize an allocation  $\allocs$ such that  bundle $\alloci[j]=\{\items[j]\}$ for $j\le k$ and  bundle $\alloci[j]=\emptyset$ for $j>k$\tcp*[l]{Allocate chores in the set $\larg{i,s_i/2}$} \label{inalg-allocate}
	Let $R= \larg{i,s_i/4}\setminus \larg{i,s_i/2}$ denote the set of remaining large chores\;
        
	\For(\tcp*[h]{Try to fill up the first $k$ bundles}\label{inalg-pre-big}){$t= k$ \KwTo $1$}
	{
		\For(\tcp*[h]{Search the largest chore could allocate}){$j= 1$  \KwTo $|R|$}
		{
		\If{$ \vai{\alloci[t]\cup R[j]}\le s_i$}{
			$\alloci[t]=\alloci[t]\cup R[j]$\;
			$R=R\setminus R[j]$\;
			}
		}
	}
	\tcp{The following for-loop is similar to Algorithm \ref{alg-approx}}
	\For(\label{inalg-after-big}){$t=k+1$ \KwTo $n$}
	{
		\For{$j= 1$  \KwTo $|R|$}
		{
			\If(\label{inalg-if-11/9}){$\vai{\alloci[t]\cup R[j]}\le \frac{5}{4}\cdot s_i$}
			{
				$\alloci[t]=\alloci[t]\cup R[j]$\;
				
			}
		}
		$R=R\setminus \alloci[t]$\;
	}
	\eIf(\tcp*[h]{Whether all large chores are allocated or not}){$R=\emptyset$}{\Return Yes\;}{\Return No\;}
\caption{Threshold testing}\label{alg-proper}
\end{algorithm}

To get a lower bound of \mmax~by a binary search,
we need monotonicity of the algorithm, i.e., threshold testing should return ``Yes" on all values larger than $\mMSi$  for any agent $i$.  
Though Example \ref{exm-no-monoton} shows monotonicity does not hold in general,
we prove that by the choice of parameters, the following monotonicity holds.
\begin{lemma}\label{lem-monotone}
For any $\alpha_i\ge\frac{11}{9}\cdot\mMSi$, if we input the threshold $\alpha_i$ to Algorithm \ref{alg-approx}, then all chores will be allocated.
\end{lemma}
\begin{proof}
Notice that, if we replace the condition of $\alpha_i$ from $\alpha_i=\frac{11}{9}\cdot\mMSi$ to $\alpha_i\ge\frac{11}{9}\cdot\mMSi$ in Lemma \ref{lem-no-less-2/9} and Lemma \ref{lem-no-larger-2/9}, the proof of those two lemmas work in the same way. This observation would directly implies this lemma.
\end{proof}

Next we prove that our threshold testing algorithm has the following monotone property. 
\begin{lemma}\label{lem-mms-ok}
If we input an index $i$ and a threshold $s_i\ge\mMSi$ to Algorithm \ref{alg-proper}, then the algorithm will return ``Yes". Therefore, when we use Algorithm \ref{alg-proper} to do binary search, we will correctly get a lower bound of \mmax.
\end{lemma}

\begin{proof}
As Algorithm \ref{alg-proper}  deal with  the chores in set $\larg{i,s_i/4}$ only,  let us assume that the chores set $\items=\larg{i,s_i/4}$ in this proof.

When the inequality $s_i\ge\mMSi$ holds, there is an allocation $\allocs^*$ such that $\vai{\alloci[j]^*}\le s_i$ for all $j$. We will compare the allocation $\allocs^*$ with the allocation generated in Algorithm  \ref{alg-proper} to show that 
 the remaining chores after line \ref{inalg-after-big} do not  greater than the chores in the allocation  $\allocs^*$.
Then we can allocate them in a way like Algorithm \ref{alg-approx}. So the threshold $s_i$ will pass the test. 

We introduce some notations for the proof. Let $H$ denote the set of chores which are allocated before line \ref{inalg-after-big}. Let $Q$ be the set of chores of bundles containing a really large chore from allocation $\allocs^*$, i.e.,   $$\left\{\bigcup \alloci[j]^*\mid \alloci[j]^*\in\allocs^*\mbox{ and }\alloci[j]^*\cap\larg{i,s_i/2}\neq\emptyset\right\}.$$

 Let $\allocs$ be the allocation generated in  Algorithm \ref{alg-proper} when we input agent $i$ and value $s_i$. Suppose that for allocation $\allocs$, the indices of the bundles  are the same as in Algorithm \ref{alg-proper}. Let  $k=\cadin{\larg{i,s_i/2}}$ be the number of really large chores. 
Notice that no two chores in the set $\larg{i,s_i/2}$ will be allocated to one bundle i.e., all these chores  are allocated separately. So we have  $H=\bigcup_{j\le k}\alloci[j]$. Without loss of generality, we can assume that the $j$-th largest chore $\items[j]\subseteq\alloci[j]^*$ for all $j\le k$. Then we have $Q=\bigcup_{j\le k}\alloci[j]^*$.

Since the set of chores $\larg{i,s_i/2}$ is contained in both set $H$ and set $Q$, we can simply let function $f(c)=c$ for any chore $c\in\larg{i,s_i/2}$. Let us consider the set of chores $Q'=Q\setminus\larg{i,s_i/2}$. 

\bigskip
{\bf Index:} If chore $c\in\alloci[j]^*$, we call $j$ the \emph{index} of chore $c\in Q'$. 
\bigskip

Recall that we have assumed that the set of all chores $\items=\larg{i,s_i/4}$. So there are no two chores in $Q'$ sharing one index.  

\begin{claim}\label{claim-assume}
Without breaking the condition $\vai{\alloci[j]^*}\le s_i$ for all $j\le k$,  we can assume that:
\begin{enumerate}[1)]
\item For any two chores $c_1,c_2\in Q'$, if $\vai{c_1}>\vai{c_2}$, then the index of chore $c_1$ greater than the index of chore $c_2$. 
\item  Indexes of chores in set $Q'$ are consecutive and end at the index $k$, i.e., the set of indexes equals the set $[j,k]$ (here $[j,k]$ is the set of  all integers between $j$ and $k$).
\end{enumerate}
\end{claim}
\begin{proof}
Recall that  the chore $\items[j]\in\alloci[j]^*$ for each $j\le k$. It means that when index becomes larger, there is more space to allocate a chore in $Q'$ without exceeding the threshold $s_i$. So we have the following observation. If  condition 1) breaks, we can swap two chores without violating \mmax~condition. If condition 2) breaks, then there is an index $j$ for a chore $c\in Q'$ and an index $d$ such that $d>j$ and there is no chore corresponding to index $d$. In this case, we can reallocate chore $c$ to bundle $\alloci[d]^*$ and it is still a \mmax~allocation. These two observations imply we can make above two assumptions.
\end{proof}

Suppose that the assumptions in Claim \ref{claim-assume} hold, we prove the following claim. 

\begin{claim}\label{claim-lessthan}
There is an injection $f:Q\rightarrow H$ such that $\vai{c}\le\vai{f(c)}$ for any chore $c\in Q$ . 
\end{claim}
\begin{proof}

Since the set of chores $\larg{i,s_i/2}$ is contained in both set $H$ and set $Q$, we can simply let function $f(c)=c$ for any chore $c\in\larg{i,s_i/2}$. Now we consider the chores $Q'=Q\setminus\larg{i,s_i/2}$.

Similarly, let $H'=H\setminus\larg{i,s_i/2}$ and call $j$ to be the index of chore $c\in H'$. Notice that, by Algorithm \ref{alg-proper}, the allocation $\allocs$ and set $H'$ satisfy these two conditions as well. For any chore $c\in Q'$, let $f(c)=c'$ where chore $c$ and chore $c'\in H'$ have a same index.

Now we prove that $\vai{c}\le\vai{f(c)}$. Let $j$ be the index for chore $c$.  When Algorithm \ref{alg-proper} allocates chore $f(c)$ to bundle $\alloci[j]$, it will choose the largest possible chore $c'$ such that $\vai{\items[j]\cup c'}\le s_i$. Deduce from our assumption for the index of chore $c$, we know that chore $c$ has not been allocated when  Algorithm \ref{alg-proper} deal with bundle $\alloci[j]$. Thus, we have $\vai{c}\le\vai{f(c)}$.
\end{proof}

Next we analysis the remaining chores from sets $\items\setminus H$ and $\items\setminus Q$.
\begin{claim}\label{claim-injective}
There is an injective function $g:\items\setminus H\rightarrow \items\setminus Q$ such that  $\vai{c}\le\vai{g(c)}$ for any chore $c\in \items\setminus H$.
\end{claim}  
\begin{proof}
By Claim \ref{claim-lessthan}, suppose that $f:Q\rightarrow H$ is an injective function such that  $\vai{c}\le\vai{f(c)}$ for any chore $c\in Q$. Let $f^{0}:\items\rightarrow \items$ be an identity function.  For any integer $k>0$, let $f^{k}(c)=f(f^{k-1}(c))$ if $f^{k-1}(c)\in Q$. We construct $g$ such that $$g(c)=f^{w}(c)\text{ s.t. } w=\min\{k\mid f^{k}(c)\in \items\setminus Q\}.$$

First, we prove that for any $c\in\items\setminus H$, the mapping $g(c)$ is well-defined. For any chore $c\in\items\setminus(Q\cup H)$, we have $f^{0}(c)=c\in\items\setminus Q$. So for those chores, we have $g(c)=c$. For any chore $c\in Q\setminus H$, we prove that there must exist a $k$ such that $f^{k}(c)\in\items\setminus Q$. Suppose that there is no such $k$. There must exist two integers $k_1<k_2$ such that  $f^{k_1}(c)=f^{k_2}(c)$. Since function $f$ is injective, we have that $$c=f^{0}(c)=f^{k_1-k_1}(c)=f^{k_2-k_1}(c)\in H.$$ Because we have assumed that $c\in\items\setminus H$, it is a contradiction. Therefore, function $g$ is well-defined.

Second, we prove that the function $g$ is injective. Suppose that there are two chores $c_1,c_2\in \items\setminus H$ such that $g(c_1)=g(c_2)$. Then there exist two integers $k_1,k_2$ such that $g(c_1)=f^{k_1}(c_1)=f^{k_2}(c_2)=g(c_2)$. Recall that $f$ is injective. If $k_1=k_2$, then $c_1=c_2$. It is impossible. Let us assume that $k_1<k_2$. Then, we have  $c_1=f^{0}(c_1)=f^{k_2-k_1}(c_2)\in H$. It is a contradiction to the range of chore $c_1$. Thus, function $g$ is injective. 

As  $\vai{c}\le\vai{f(c)}$ for any chore $c$, we have $\vai{c}\le \vai{f^{w_i}(c)}=\vai{g(c)}$.
\end{proof}

The set $\items\setminus Q$ is exactly the set $\bigcup_{j>k}\alloci[j]^*$. By Claim \ref{claim-injective}, we are able to allocate set $\items\setminus H$ into $n-k$ bundles such that no bundle with valuation more than $s_i$. Then by Lemma \ref{lem-monotone}, Algorithm \ref{alg-proper} will allocate all remaining chores by the threshold $s_i$. So Algorithm \ref{alg-proper} will return ``Yes" on the threshold $s_i$.

\end{proof}

\subsection{Efficient approximation algorithm}
Now we present the following approximation algorithm for \amMS{5/4} allocation.

\begin{algorithm}[H]
\KwIn{An IDO instance $\inst$}
\KwOut{An allocation  $\allocs$ of $\inst$}
\BlankLine

         \For{$i\in\agents$}
         {
         	Let $l_i=\max\left\{\vai{\items}/n,\vai{\items[1]}\right\}$ and $r_i=2l_i$.\;\label{inalg-upper-bound}
         With lower bound $l_i$ and upper bound $r_i$, binary search for the minimal integer  $s_i$  such that Algorithm \ref{alg-proper} return Yes on input $(\val,s_i)$\;

	
         }
         Run Algorithm \ref{alg-approx} on input $\inst,\left(\frac{5}{4}\cdot s_1,\dots,\frac{5}{4}\cdot s_n\right)$\;\label{inalg-call-approx}
         Let $\allocs$ be the output of Algorithm \ref{alg-approx}\;
	
\Return{$\allocs$}
\caption{Polynomial time algorithm for 5/4 approximation}\label{alg-poly}
\end{algorithm}

{\bf Remark:}  The value $r_i=2l_i$ from line \ref{inalg-upper-bound} is an upper bound of $\mMSi$. This can be deduced from the 2 approximation algorithm of the job scheduling problem~\cite{graham1966bounds}.

To analyze the correctness of Algorithm \ref{alg-poly}, we  focus on the agent who gets the last bundle, which is denoted by $\last$. 
First we argue that no small chores left unallocated. 
\begin{lemma}\label{lem-no-less-1/4}
If the threshold $s_{\last}\ge \vai[\last]{\items}/n$, then after Algorithm \ref{alg-poly} terminates, all chores $c\in \smal{\last,s_{\last}/4}$ will be allocated.
\end {lemma}
\begin{proof}
By the same argument as Lemma \ref{lem-no-less-2/9}, if there is a chore $c\in \smal{\last,s_{\last}/4}$ left, then every bundle is greater than $s_{\last}$. In the algorithm, we have set a lower bound $s_{\last}\ge \vai[\last]{\items}/n$ for the binary search. If the value of every bundle is greater than $s_{\last}$, the sum of valuations of all bundles will grater than $ \vai[\last]{\items}$. This is impossible. Thus we prove the lemma.
\end{proof}

Next we analyze what happened to the large chores in the allocation. The idea of the proof is similar to what we have done for Lemma \ref{lem-no-larger-2/9}. The key difference is that the benchmark allocation is the one generated by the threshold testing, but not a \mmax~allocation like Lemma \ref{lem-no-larger-2/9}. 
We compare the output of Algorithm \ref{alg-poly} with the benchmark allocation to argue that all chores will be allocated.

\begin{lemma}\label{lem-no-grater-1/4}
After Algorithm \ref{alg-poly} terminates, all chores $c\in \larg{\last,s_{\last}/4}$ will be allocated.
\end{lemma}
\begin{proof}[Proof sketch]
Please see Appendix \ref{apd-lemma-no-1/4} for the full proof.  Here we give a high level idea of the proof. After binary search, we have a threshold value $s_{\last}$ for the last agent $\last$, which is  a lower bound of $\mMSi[\last]$. Let $\bundles$ be the allocation that the threshold testing algorithm generated inside itself upon inputting agent $\last$ and value $s_{\last}$. Let $\allocs$ be the allocation returned by Algorithm \ref{alg-poly}. We will try to compare allocation $\allocs$ and allocation $\bundles$ from agent $\last$'s perspective.  

When comparing these two allocations, we observe that, for any $k$, the set of the first $k$ bundles of allocation $\allocs$  ($\bigcup_{i\le k}\alloci$) always contains larger and more chores than allocation $\bundles$ ($\bigcup_{i\le k}\bundlei$). Particularly, we prove by induction that we can maintain an injective mapping from chore set  $\bigcup_{ k<i\le n}\bundlei$ to set $\bigcup_{k< i\le n}\alloci$ such that each chore $c$ is mapped to a chore with valuation no less than $c$. 
\end{proof}

With all these lemmas, now we can prove the following theorem. 

\begin{theorem}\label{thm-poly}
For integer valuations, Algorithm \ref{alg-poly} will output a 5/4 approximation maximin share allocation in $O(nm\log m+n^2)$ time.
\end{theorem}
\begin{proof}
The correctness of Algorithm \ref{alg-poly} is directly implied by Lemma \ref{lem-no-less-1/4} and Lemma \ref{lem-no-grater-1/4}. We only need to analyze the time complexity of our algorithm. The running time of the threshold testing is about $O(m\log m +n)$. As valuations are all integers, the number of iterations of binary search depends on the length of the representation of each integer (usually it can be considered a constant). 
 And we repeat the binary search $n$ times for each agent respectively. Finally, Algorithm \ref{alg-approx} terminates in time $O(nm\log m)$. Combining all these, the total running time of our algorithm is $O(nm\log m+n^2)$.
\end{proof}

\section{Application to the job scheduling problem}\label{sec-special}

The job scheduling problem is one of the fundamental discrete optimization problems. Here we particularly consider the model that minimizes the execution time for scheduling $m$ jobs on $n$ identical machines. It could be viewed as a special case of chore allocation, where all agents have the same valuation. From this perspective, we show how to apply our algorithmic framework to this problem.

The problem of job scheduling is proved to be NP-hard in~\cite{DBLP:books/fm/GareyJ79}. And later, the polynomial time approximation scheme (PTAS) for this problem was discovered and developed~\cite{alon1998approximation,jansen2010eptas,jansen2016closing}.
To the best of our knowledge, except those PTASs, there is no algorithm approximating optimal better than 4/3. From our algorithmic framework, we demonstrate an algorithm, which is simpler  and more efficient than the best PTAS, and achieves a better approximation ratio than other heuristics. 

\begin{theorem}\label{thm-schedul}
When all valuations are integers, an 11/9 approximation of optimal scheduling can be found in $O(m\log m+n)$ time. 
\end{theorem}
\begin{proof}

We first describe the algorithm for this problem. The algorithm  is similar to Algorithm \ref{alg-poly}, except for two modifications: 
\begin{itemize}
\item Change the threshold testing algorithm in binary search from Algorithm \ref{alg-proper} to the naive test which is introduced in section \ref{subsec-failure}.
\item Delete the for-loop, and compute one proper threshold from the valuation function.
\end{itemize}

As all agents share the same valuation function, the threshold obtained by the naive test obviously can apply to all agents. Then we can get an allocation such that no bundle exceeds the threshold. And by Lemma \ref{lem-monotone} , the threshold from binary search is not greater than $\frac{11}{9}\cdot\mMS$.

The time complexity of one testing is $O(m\log m+n)$. The number of iterations of binary search depends on the length of the representation of each integer (usually it can be considered a constant). 
Finally, Algorithm \ref{alg-approx} will be executed one more time. Thus, the complexity of our algorithm is $O(m\log m+n)$
\end{proof}

\section{Discussion}\label{sec-dis}

At the first glance, our algorithm is quite similar to FFD algorithm for the bin packing problem. However, technically they are two different problems. 1) The bin packing problem is fixing the size of each bin and then try to minimize the number of bins. Our problem can be viewed as fixing the number of bins but try to minimize the size of each bin. 2) The optimal approximation of FFD algorithm for the bin packing problem is $\frac{11}{9}\cdot OPT+\frac{6}{9}$. To the best of our knowledge, there is no reduction between the approximation ratio of these two problems.  3) The lower bound Example \ref{exm-lower-bound}  does not make sense to the bin packing problem, and the lower bound example of bin packing problem does not make sense to our problem. 


For the analysis of our algorithmic framework, there is a gap between the lower and the upper bound of the approximation ratio. It will be interesting to close this gap. 
Another interesting direction is that how to convert the existence result from our algorithmic framework into an efficient algorithm. Suppose that $\alpha$ is the best approximation ratio of existence of our algorithmic framework. By combining a PTAS for job scheduling ~\cite{alon1998approximation}, we can have a PTAS for $\alpha+\epsilon$ approximation of \mmax~allocation. Nevertheless, such PTAS can hardly to be considered as efficient in practical use.  In section \ref{sec-poly}, we show one way to get a polynomial time algorithm for 5/4 approximation from our existence result. It is still possible to design an efficient algorithm for  $\alpha$-approximation.   

In section \ref{sec-special}, we try to explore the power of our algorithmic framework on the job scheduling problem.   It may worth to exploring more on the relationship of the algorithmic framework with other scheduling problems. 

\section*{Acknowledgement}
Xin Huang is supported in part at the Technion by an Aly Kaufman Fellowship. Pinyan Lu is supported by Science and Technology Innovation 2030 -- ``New Generation of Artificial Intelligence" Major Project No.(2018AAA0100903), NSFC grant 61922052 and 61932002, Innovation Program of Shanghai Municipal Education Commission, Program for Innovative Research Team of Shanghai University of Finance and Economics, and the Fundamental Research Funds for the Central Universities.

We thank Prof. Xiaohui Bei for helpful discussion on the subject. Thank Prof. Inbal Talgam-Cohen and Yotam Gafni for providing useful advice on writing. Part of this work was done while the author Xin Huang was visiting the Institute for Theoretical Computer Science at Shanghai University of Finance and Economics.

\bibliographystyle{acm}

\bibliography{ref}


\appendix
\section{Reduction from general to identical ordinary preference}\label{apd-sec-reduce}
Here we formally state a reduction technique which is introduced by Bouveret and Lema{\^\i}tre~\cite{bouveret2016characterizing}. By this reduction, we only need to take care of instances such that all agents share the same ordinary preferences.
And this technique has been successfully applied in the work \cite{barman2017approximation} to simplify the proof of approximation of maximin share allocation for goods.

We first introduce a concept of ordered instance.
\begin{definition}[Ordered instance]
Given any instance $\inst=\langle\agents,\items,\va\rangle$, the corresponding \emph{ordered instance} of ~$\inst$ is denoted as $\inst^*=\langle\agents^*,\items^*,\va^*\rangle$, where:
\begin{itemize}
\item $\agents^*=\agents$, and $|\items^*|=|\items|$
\item For each agent $i$ and chore $c^*_j\in\items^*$, we have $\val^*(c^*_j)=\vai{\items[i,j]}$
\end{itemize}
\end{definition}

The whole reduction relies on the following algorithm.

\begin{algorithm}[H]
\KwIn{$\inst$, its ordered instance $\inst^*$, and an allocation $\allocs^*$ of $\inst^*$}
\KwOut{An allocation  $\allocs$ of $\inst$}
\BlankLine
	Let $\allocs=(\emptyset,\dots,\emptyset)$ {\bf and} $T=\items$\;
	\For(\tcp*[f]{From the smallest chore to the largest}){$j= m$ \KwTo $1$}
	{
		Let $i$ be the agent such that $\items^*[j]\in\alloci^*$\;
		Find  $c= \argmin_{c'\in T} \vai{c'}$\;\label{inalg-find}
		$\alloci=\alloci\cup c$ {\bf and} $T=T\setminus c$\;
	}
\Return{$\allocs$}
\caption{Reduction}\label{alg-reduce}
\end{algorithm}

We have the following observation for this algorithm. 
\begin{lemma}\label{lem-reduce}
Given an instance $\inst$, its ordered instance $\inst^*$ and an allocation $\allocs^*$ of $\inst^*$,  Algorithm \ref{alg-reduce} will output an allocation $\allocs$ such that $\vai{\alloci}\le\vai{\alloci^*}$ for each agent $i$.
\end{lemma}
\begin{proof}
To prove this lemma, it is sufficient to prove that in line \ref{inalg-find} of the algorithm, the chore $c$ satisfies $\vai{c}\le\vai{\items^*[j]}$. Because if this is true, after adding up all these inequalities, we can get $\vai{\alloci}\le\vai{\alloci^*}$.

Look at the for-loop in  Algorithm \ref{alg-reduce}. In the round of $j$, there are  $j$ chores left in $T$, i.e., $|T|=j$.  We have $\vai{c}=\vai{T[i,j]}\le\vai{\items[i,j]}=\vai{\items^*[j]}$.
\end{proof}

Since the set of chores of ordered instance $\inst^*$ is the same as $\inst$, the \mmax~of each agent $i$ is the same in both instances.

\begin{proof}[Proof of Lemma \ref{lem-general}]
Given such an algorithm $G$, we can construct an algorithm for the general instance as following.

Given a chore division instance $\inst$, we can construct its ordered instance $\inst^*$ in $O(nm\log m)$ time, which is by sorting algorithm for all agents.   And then we run algorithm $G$ on the instance $\inst^*$ to get an \amMS{\alpha} allocation $\allocs^*$ for $\inst^*$. Then we run Algorithm \ref{alg-reduce} on $\allocs^*$ to get  an allocation $\allocs$ of instance $\inst$. The Algorithm \ref{alg-reduce} can be done in $O(nm)$ time. And with lemma \ref{lem-reduce}, we know that $\allocs$ is an \amMS{\alpha} allocation of instance $\inst$.
\end{proof}


\section{Complementary Proofs of Lemma \ref{lem-no-larger-2/9}}\label{apd-no-2/9}

\begin{proof}[{\bf Proof of Claim \ref{claim-4chores}}]
For the case $\cadin{\bundlei[k+1]^{(k)}}=1$, it is trivial. For the case  $\cadin{\bundlei[k+1]^{(k)}}=2$, we notice that Algorithm \ref{alg-approx}  can at least allocate chore $\bundlei[k+1]^{(k)}[2]$ to bundle $\alloci[k+1]^*$ as the second chore. This would imply the cardinality  $\cadin{\alloci[k+1]^*}\ge 2$ and the valuation $\vai[\last]{\bundlei[k+1]^{(k)}[2]}\le\vai[\last]{\alloci[k+1]^*[2]}$.

For the case $\cadin{\bundlei[k+1]^{(k)}}=4$, first we prove that the inequality  $U\left(\bundlei[k+1]^{(k)}\right)<5/9$ is true. Since bundle $\bundlei[k+1]^{(k)}$ is good, we have the condition $U\left(\bundlei[k+1]^{(k)}\right)<5/9$ or  $\vai[\last]{\bundlei[k+1]^{(k)}}\le1$.
If the valuation $\vai[\last]{\bundlei[k+1]^{(k)}}\le1$, then as  any chores in the bundle is greater than 2/9,  we have $$U\left(\bundlei[k+1]^{(k)}\right)=1-\text{value of the other two chores}<1-2\cdot 2/9=5/9.$$

Now we prove the inequality $$3\cdot \vai[\last]{\bundlei[k+1]^{(k)}[1]}+\vai[\last]{\bundlei[k+1]^{(k)}[4]}\le11/9.$$
Suppose that $3\cdot \vai[\last]{\bundlei[k+1]^{(k)}[1]}+\vai[\last]{\bundlei[k+1]^{(k)}[4]}>11/9$. Let $x$ be the value of $\vai[\last]{\bundlei[k+1]^{(k)}[4]}$. Then we have
\begin{align*}
11/9&<3\cdot \vai[\last]{\bundlei[k+1]^{(k)}[1]}+\vai[\last]{\bundlei[k+1]^{(k)}[4]}\\
&<3\cdot(5/9-x)+x \tag{by $U\left(\bundlei[k+1]^{(k)}\right)<5/9$}
\end{align*}
This implies $x<2/9$. This contradicts our assumption that $c\in \larg{\last,2/9}$.

Notice that $\vai[\last]{\bundlei[k+1]^{(k)}[1]}$ is the value of largest chore among the remaining. So the inequality suggests that even if we allocate 3 largest chores, there is still a space for the chore $\bundlei[k+1]^{(k)}[4]$. So the first 3 chores being allocated to bundle  $\alloci[k+1]^*$ must be the largest 3 chores. So  $\vai[\last]{\bundlei[k+1]^{(k)}[j]}\le\vai[\last]{\alloci[k+1]^*[j]}$ for $j\le3$.

As we have $\cadin{\bundlei[k+1]^{(k)}}=4$, the chore $\bundlei[k+1]^{(k)}[4]$ would not be one of first 3 chores allocated to bundle  $\alloci[k+1]^*$. Since there is a space for  the chore $\bundlei[k+1]^{(k)}[4]$, the forth chore being allocated to bundle  $\alloci[k+1]^*$ should at least as large as  the chore $\bundlei[k+1]^{(k)}[4]$.  Combining all these, we have $\cadin{\alloci[k+1]^*}\ge 4$ and $\vai[\last]{\bundlei[k+1]^{(k)}[j]}\le\vai[\last]{\alloci[k+1]^*[j]}$ for $j\le4$.
\end{proof}

\bigskip

\begin{proof}[{\bf Proof of Claim \ref{claim-two-one}}]
We first prove two useful observations.
\begin{observation}\label{obser-2is-largest}
The chore $\alloci[k+1]^*[2]$ must be the largest chore in the bundle $\bundlei[j]^{(k)}$.

\begin{proof}
Because $|\bundlei[k+1]^{(k)}|=3$ and it is good, so $\vai[\last]{\bundlei[k+1]^{(k)}}\le1$, which means $\bundlei[k+1]^{(k)}[1]< 1-2/9-2/9=5/9$. This implies all remaining chores are less than $5/9$. Therefore, add the second largest chore, it is still less than $11/9$. We have  $\alloci[k+1]^*[2]$ must be the largest chore in the bundle $\bundlei[j]^{(k)}$.
\end{proof}
\end{observation}

\begin{observation}\label{obser-greater-4/9}
For the valuation of chore $\alloci[k+1]^*[2]$, we have  $\vai[\last]{\alloci[k+1]^*[2]}>4/9$.
\begin{proof}
As Inequality \ref{ineq-2-3} does not hold, we have
\begin{align*}
\vai[\last]{\left\{\alloci[k+1]^*[1],\alloci[k+1]^*[2], \bundlei[k+1]^{(k)}[3]\right\}}&> 11/9\\
&\ge2/9+\vai[\last]{\bundlei[k+1]^{(k)}}\\
&=2/9+\vai[\last]{\left\{\bundlei[k+1]^{(k)}[1],\bundlei[k+1]^{(k)}[2], \bundlei[k+1]^{(k)}[3]\right\}}
\end{align*}
By the assumption that  chore $\alloci[k+1]^*[1]=\bundlei[k+1]^{(k)}[1]$, we get $$\vai[\last]{\alloci[k+1]^*[2]}>2/9+\vai[\last]{\bundlei[k+1]^{(k)}[2]}>4/9.$$
\end{proof}

\end{observation}

Now we prove the statement by case analysis on the number of chores in the  bundle $\bundlei[j]^{(k)}$.

For the case that bundle $\bundlei[j]^{(k)}$ only contains the chore $\alloci[k+1]^*[2]$, i.e., $\cadin{\bundlei[j]^{(k)}}=1$, it is easy to see bundle $\bundlei[j]^{(k+1)}$ is good after replacing the chore $\alloci[k+1]^*[2]$ with chores $\bundlei[k+1]^{(k)}[2]$ and $\bundlei[k+1]^{(k)}[3]$.

If bundle $\bundlei[j]^{(k)}$ contains two chores, i.e., $\cadin{\bundlei[j]^{(k)}}=2$, we prove the inequality $\vai[\last]{\bundlei[j]^{(k+1)}}\le 1.$ First we have the following inequality
\begin{align*}
\vai[\last]{\left\{\bundlei[k+1]^{(k)}[2], \bundlei[k+1]^{(k)}[3]\right\}}&\le1-\vai[\last]{\bundlei[k+1]^{(k)}[1]}\\
&\le1-\vai[\last]{\alloci[k+1]^*[2]}.
\end{align*}
For the valuation of bundle $\bundlei[j]^{(k+1)}$, we have
\begin{align*}
\vai[\last]{\bundlei[j]^{(k+1)}}& = \vai[\last]{\left(\bundlei[j]^{(k)}\setminus \alloci[k+1]^*[2]\right)\bigcup\left(\left\{\bundlei[k+1]^{(k)}[2], \bundlei[k+1]^{(k)}[3]\right\}\right) }\\
&\le\vai[\last]{\alloci[k+1]^*[2]}+\vai[\last]{\left\{\bundlei[k+1]^{(k)}[2], \bundlei[k+1]^{(k)}[3]\right\}}\\
&\le \vai[\last]{\alloci[k+1]^*[2]}+1-\vai[\last]{\alloci[k+1]^*[2]}\\
&=1.
\end{align*}
The first inequality is due to the chore $\alloci[k+1]^*[2]$ is no less than the other chore in bundle $\bundlei[j]^{(k)}$. The second inequality just follows the above inequality.
Thus, bundle $\bundlei[j]^{(k+1)}$ is  still good.

Now let us look at the case $\cadin{\bundlei[j]^{(k)}}=3$. Since bundle $\bundlei[j]^{(k)}$ is good, we have $\vai[\last]{\bundlei[j]^{(k)}}\le1$.  By Observation \ref{obser-greater-4/9}, we have  $\vai[\last]{\alloci[k+1]^*[2]}>4/9$.
Therefore, except the chore $\alloci[k+1]^*[2]$, the valuation of the remaining two chores in bundle $\bundlei[j]^{(k)}$ is $$\vai[\last]{\bundlei[j]^{(k)}\setminus \alloci[k+1]^*[2]}\le 1-\vai[\last]{\alloci[k+1]^*[2]}<5/9.$$
As the valuation $\vai[\last]{\bundlei[k+1]^{(k)}[1]}\ge\vai[\last]{\alloci[k+1]^*[2]}>4/9$ and bundle $\bundlei[k+1]^{(k)}$ is good,   we have $$\vai[\last]{\left\{\bundlei[k+1]^{(k)}[2], \bundlei[k+1]^{(k)}[3]\right\}}\le1-\vai[\last]{\bundlei[k+1]^{(k)}[1]}<5/9.$$
For the bundle $\bundlei[j]^{(k+1)}$, the largest chore is either in the set  $\bundlei[j]^{(k)}\setminus \alloci[k+1]^*[2]$ or in the set $\bundlei[k+1]^{(k)}[2]\cup \bundlei[k+1]^{(k)}[3].$ Therefore, the valuation $$U\left(\bundlei[j]^{(k+1)}\right)\le\max\left\{\vai[\last]{\bundlei[j]^{(k)}\setminus \alloci[k+1]^*[2]},\vai[\last]{\left\{\bundlei[k+1]^{(k)}[2],\bundlei[k+1]^{(k)}[3]\right\}}\right\}<5/9.$$
It is easy to see, in this case, the cardinality $\cadin{\bundlei[j]^{(k+1)}}=4$. Thus, bundle $\bundlei[j]^{(k+1)}$ is good.

It is impossible that $\cadin{\bundlei[j]^{(k)}}=4$. Since bundle $\bundlei[j]^{(k)}$ is good, if $\cadin{\bundlei[j]^{(k)}}=4$, then we have $U\left(\bundlei[j]^{(k)}\right)<5/9.$ As bundle $\bundlei[j]^{(k)}\subseteq \larg{\last,2/9}$, it implies that the valuation of the largest chore $$\vai[\last]{\alloci[k+1]^*[2]}=\vai[\last]{\bundlei[j]^{(k)}[1]}\le U\left(\bundlei[j]^{(k)}\right)-2/9<1/3.$$  This contradicts to  Observation \ref{obser-greater-4/9}, i.e., $\vai[\last]{\alloci[k+1]^*[2]}>4/9$. So it is impossible.
\end{proof}

\section{Proof of Lemma \ref{lem-no-grater-1/4}}\label{apd-lemma-no-1/4}

 Let $\last$ be the agent who gets the last bundle in Algorithm \ref{alg-poly} and $s_{\last}$ be the threshold found by binary search. Let $\bundles$ be the allocation generated by threshold testing when input agent $\last$ and threshold $s_{\last}$. For the ordering of bundles from allocation $\bundles$, it is the same as the ordering generated by Algorithm \ref{alg-proper}.  It will serve as a benchmark allocation to show that not many chores will be left for agent $\last$.  
 
 Let $\allocs$ be the allocation returned by Algorithm \ref{alg-poly}. For the ordering of the allocation $\allocs$, we assume that bundle $\alloci[1]$ is the first bundle generated by Algorithm \ref{alg-approx} in line \ref{inalg-call-approx} of Algorithm \ref{alg-poly}, and bundle $\alloci[2]$ is the second bundle   generated by Algorithm \ref{alg-approx} etc. In this ordering, agent $\last$ gets the bundle $\alloci[n]$. It may be a little confusing here. But we can reorder agents so that agent $i$ gets the bundle $\alloci$. 
 
 We only consider  large chores in this proof.   Let bundle $\alloci^*=\alloci\cap \larg{\last, s_{\last}/4}$ consist of large chores from bundle $\alloci$. 
 We will build an inductive argument such that, for each $k$, except first $k$ bundles, the remaining chores of allocation $\allocs^*$ are not much comparing with benchmark allocation $\bundles$.  Here we introduce some notations to formally prove the induction. Let $D_k= \larg{\last, s_{\last}/4}\setminus\left( \bigcup_{i< k}\alloci^*\right)$ be the set of  large chores that have not been allocated in first $k-1$ bundles. 
Let chore set  $P_k=\bigcup_{i\ge k}\bundlei$ contain all remaining chores of the benchmark allocation. We will prove the following property which directly imply that no chore $c\in \larg{\last,s_{\last}/4}$ left.

\bigskip
{\bf Property $T$:} For each $1\le k\le n$, there exists an injective function $f_k:D_k\rightarrow P_k$ such that $\vai[\last]{c}\le \vai[\last]{f_k(c)}$ for all  $c\in D_k$.
\bigskip

When $k=n$, the property $T$ implies the total value of large chores for the last bundle is not grater than $\vai[\last]{\bundlei[n]}$, which is less than $\frac54\cdot s_{\last}$. Therefore, all large chores will be allocated. 

We will prove the property $T$ by induction. When $k=1$, the identical mapping will work. Now we prove that if the statement holds for  $k$, then we can construct a suitable mapping for  $k+1$. 

From function $f_k$ to function $f_{k+1}$, we should do two things: 
\begin{enumerate}[1)] 
\item We will try to do some swap operations on function $f_k$ so that no chore in set $D_{k+1}$ would map to a chore in bundle $\bundlei[k]$. Then we can shrink both the domain and range from function $f_k$ to function $f_{k+1}$.
\item Meanwhile, we need to make sure that a chore is always mapped to a larger chore. 
\end{enumerate}

We will give a construction to satisfy above two conditions. 
 First we introduce a swap operation for an injective function. Given an injective function $f:D\rightarrow P$, chore $d\in D$ and chore $p\in P$, the operation $f^*=\text{Swap}(f,d,p)$ is defined as
\begin{equation*}
 f^*(c) =
  \begin{cases}
   f(c) & \text{for } c\neq d \text{ or } f^{-1}(p)  \\
   p       & \text{if } c=d  \\
   f(d)  & \text{if } c=f^{-1}(p)
  \end{cases}
\end{equation*}
 Intuitively, this operation just swaps the mapping images of chore $d$ and chore $f^{-1}(p)$. If $f^{-1}(p)=\emptyset$, then this operation just change the image of $d$ to $p$.

The detail of the construction is as following. Let $q=\min\{|\alloci[k]^*|,|\bundlei[k]|\}$ and $g_0=f_k$ be an injective function. For $1\le t\le q$, we iteratively construct a mapping by swap operation as following $$g_t=\text{Swap}(g_{t-1},\alloci[k]^*[t],\bundlei[k][t]).$$  Then let  function $f_{k+1}(c)=g_q(c)$ for all $c\in D_{k+1}$.


Now we show that such construction satisfies the Property $T$. By Claim \ref{claim-no-more} and the construction of mapping function $g_q$, we have $g_q(D_{k+1})\subseteq P_{k+1}$, which means we can shrink the domain and range. And by Claim \ref{claim-not-bad}, we have $\vai[\last]{c}\le \vai[\last]{g_q(c)}$ for $c\in D_{k+1}$. Combining these, our construction of $f_{k+1}$ is valid.

Therefore, the total value of remaining large chores for  the last bundle is not grater than $\vai[\last]{\bundlei[n]}$, which is not greater than $\frac{5}{4}\cdot s_{\last}$. This completes our proof.

\begin{claim}\label{claim-all-large}
For each $k$ and $1\le j<|\bundlei[k]|$, we have the equation such that chore $\bundlei[k][j]=P_k[j]$, i.e., except the smallest chore in bundle $\bundlei[k]$, other chores are the largest chores in the remaining. 
\end{claim}
\begin{proof}
In threshold testing, the bundle is constructed in two stages. We first consider the bundle $\bundlei[k]$  is constructed before the for-loop in line \ref{inalg-after-big}. We have $|\bundlei[k]|\le 2$ and $\bundlei[k][1]=P_k[1]$.  So all bundles, which are constructed in this stage, satisfy the statement.

Now consider the bundle $\bundlei[k]$  is constructed in the for-loop of line \ref{inalg-after-big}. The bundle $\bundlei[k]$ is constructed by adding chores one by one from largest to smallest. Suppose that $j^*<|\bundlei[k]|$ is the first index that $\bundlei[k][j^*]\neq P_k[j^*]$. Notice that in this stage, for any chore $c$ we have $s_{\last}/4<\vai[\last]{c}\le s_{\last}/2.$  This implies  
\begin{align*}
\vai[\last]{\left\{\bundlei[k][j^*],\bundlei[k][j^*+1]\right\}}&>\frac{s_{\last}}{4}+\frac{s_{\last}}{4}\\
&>\vai[\last]{P_k[j^*]}.
\end{align*} 
So there is enough space for chore $P_k[j^*]$ to put into bundle $\bundlei[k]$. According to the threshold testing algorithm, chore $P_k[j^*]$ should be put into bundle $\bundlei[k]$
 This is a contradiction. Therefore, we have $\bundlei[k][j]=P_k[j]$ for $1\le j<|\bundlei[k]|$.

\end{proof}

 \begin{claim}\label{claim-not-bad}
 Recall that $q=\min\{|\alloci[k]^*|,|\bundlei[k]|\}$. For mapping $g_q$, we have $\vai[\last]{c}\le \vai[\last]{g_q(c)}$ for any chore $c\in D_k\setminus \alloci[k]^*$.
 \end{claim}
 \begin{proof}
We first prove that, for $1\le t<q$, the statement $$\forall c\in D_k, \vai[\last]{c}\le \vai[\last]{g_t(c)}$$ is true. We prove it by induction. Suppose that the statement $\forall c\in D_k, \vai[\last]{c}\le \vai[\last]{g_{t-1}(c)}$ is true.
In each swap operation, we map chore $\alloci[k]^*[t]$ to chore $\bundlei[k][t]$ and chore $g_{t-1}^{-1}(\bundlei[k][t])$ to chore $g_{t-1}(\alloci[k]^*[t])$. 
 Since bundle $\alloci[k]^*\subseteq P_k$, by Claim \ref{claim-all-large}, we have $$\vai[\last]{\alloci[k]^*[t]}\le\vai[\last]{P_k[t]}= \vai[\last]{\bundlei[k][t]}.$$

Now we will prove $\vai[\last]{g_{t-1}^{-1}(\bundlei[k][t])}\le \vai[\last]{g_{t-1}(\alloci[k]^*[t])}$. If  chore $g_{t-1}^{-1}(\bundlei[k][t])=\alloci[k]^*[t]$, then this is the above case.  If chore $g_{t-1}^{-1}(\bundlei[k][t])\neq\alloci[k]^*[t])$, recall  the construction of bundle $\alloci[k]$. When we allocate the  chore $\alloci[k]^*[t]$ to the bundle $\alloci[k]$, if there is enough space to allocate the chore $g_{t-1}^{-1}(\bundlei[k][t])$ but algorithm does not do that, it means $$\vai[\last]{g_{t-1}^{-1}(\bundlei[k][t])}\le\vai[\last]{\alloci[k]}\le \vai[\last]{g_{t-1}(\alloci[k]^*[t])}.$$ So we only need to prove that there is enough space for chore $g_{t-1}^{-1}(\bundlei[k][t])$.
Notice that 
\begin{align*}
\vai[\last]{\left(\cup_{h<t}\alloci[k]^*[h]\right)\cup g_{t-1}^{-1}(\bundlei[k][t])}&\le\vai[\last]{\left(\cup_{h<t} g_{t-1}(\alloci[k]^*[h])\right)\cup \bundlei[k][t]}\\
&\le\vai[\last]{\cup_{h\le t}P_k[h]}\\
&=\vai[\last]{\cup_{h\le t}\bundlei[k][h]}\\
&<\frac{5}{4}\cdot s_{\last}
\end{align*}
The equality in the third line is by Claim \ref{claim-all-large}.   This means that there is enough space for chore $\bundlei[k][t]$ to be allocated when algorithm allocates the  chore $\alloci[k]^*[t]$ to the bundle $\alloci[k]$. As $\vai[\last]{g_{t-1}^{-1}(\bundlei[k][t])}\le\vai[\last]{\bundlei[k][t]}$, there is enough space for chore $g_{t-1}^{-1}(\bundlei[k][t])$. Therefore, we have $\vai[\last]{g_{t-1}^{-1}(\bundlei[k][t])}\le \vai[\last]{g_{t-1}(\alloci[k]^*[t])}$.

For $g_q$, what happened to chore $\alloci[k]^*[q]$ is not our concern.  We only need to argue that $$\vai[\last]{g_{t-1}^{-1}(\bundlei[k][q])}\le \vai[\last]{g_{t-1}(\alloci[k]^*[q])}.$$ The remaining proof is similar to the above case. Recall  the construction of bundle $\alloci[k]$, when we allocate the  chore $\alloci[k]^*[q]$ to the bundle $\alloci[k]$,  there is enough space to allocate the chore $g_{t-1}^{-1}(\bundlei[k][q])$. Because
\begin{align*}
\vai[\last]{\left(\cup_{h<q}\alloci[k]^*[h]\right)\cup g_{t-1}^{-1}(\bundlei[k][q])}&\le\vai[\last]{\left(\cup_{h<q} f_k(\alloci[k]^*[h])\right)\cup\bundlei[k][q]}\\
&\le\vai[\last]{\left(\cup_{h<q}P_k[h]\right)\cup\bundlei[k][q]}\\
&\le\vai[\last]{\bundlei[k]}\\
&<\frac{5}{4}\cdot s_{\last}
\end{align*}


 \end{proof}

\begin{claim}\label{claim-no-more}
If the cardinality $|\alloci[k]^*|<|\bundlei[k]|$, then for $|\alloci[k]^*|<j\le|\bundlei[k]|$, we have $g_q^{-1}(\bundlei[k][j])=\emptyset$.
\end{claim}
\begin{proof}
For the valuation of bundle $\alloci[k]^*$, we have
\begin{align*}
\vai[\last]{\alloci[k]^*}&\le \vai[\last]{\bigcup_{1\le t\le|\alloci[k]|}\bundlei[k][t]}\\
&\le \frac{11}{9}\cdot s_{\last}-\vai[\last]{\bundlei[k][j]}\\
&\le \frac{5}{4}\cdot s_{\last}-\vai[\last]{\bundlei[k][j]},
\end{align*}
where $|\alloci[k]^*|<j\le|\bundlei[k]|$. The first inequality is implied by Claim \ref{claim-all-large} and Claim \ref{claim-not-bad}. This means that if there is a chore mapping to  chore $\bundlei[k][j]$, then this chore can be added into $\alloci[k]^*$.
\end{proof}

\end{document}